\documentclass[conference,letterpaper]{IEEEtran}

\usepackage[utf8]{inputenc} 
\usepackage[T1]{fontenc}
\usepackage{url}
\usepackage{ifthen}
\usepackage{cite}
\usepackage[cmex10]{amsmath} 

\usepackage{amssymb,amsfonts,amsthm, color}
\usepackage{graphicx}
\usepackage{textcomp}
\usepackage{xcolor}
\usepackage{enumitem}
\usepackage{algorithm2e}
\SetKwInput{KwData}{Encoder}
\SetKwInput{KwResult}{The result}
\usepackage{float}
\newtheorem*{theorem*}{Theorem}
\newtheorem{theorem}{Theorem} 
\newtheorem{corollary}{Corollary}

\newtheorem{lemma}{Lemma} 

\newtheorem{claim}{Claim}
\usepackage[utf8]{inputenc}
\usepackage{xcolor}
\usepackage[noend]{algpseudocode}
\usepackage{mathtools}
\usepackage{cuted}
\usepackage{caption}

\newcommand{\C}{{\mathcal C}}

\newcommand{\E}{{\mathcal E}}

\newcommand{\cO}{{\mathcal O}}

\newcommand{\R}{{\mathcal R}}

\renewcommand{\S}{{\mathcal S}}

\renewcommand{\O}{{\mathcal O}}

\newcommand{\ba}{{\boldsymbol a}}

\newcommand{\bw}{{\boldsymbol{w}}}

\begin{document}
\title{Mass Error-Correction Codes for Polymer-Based Data Storage}

\author{%
  \IEEEauthorblockN{Ryan Gabrys}
  \IEEEauthorblockA{University of California, San Diego\\
                    SPAWAR, San Diego\\
                    San Diego, CA, 92093, USA\\
                    ryan.gabrys@gmail.com}
  \and
  \IEEEauthorblockN{Srilakshmi Pattabiraman and Olgica Milenkovic}
  \IEEEauthorblockA{The Department
of Electrical and Computer Engineering \& CSL\\
                    University of Illinois, Urbana-Champaign\\ 
                    Urbana, IL, 61801\\
                    Email: \{sp16, milenkov \}@illinois.edu}
}

\maketitle

\begin{abstract}
We consider the problem of correcting mass readout errors in information encoded in binary polymer strings. 
Our work builds on results for string reconstruction problems using composition multisets~\cite{acharya2014string} and the unique string 
reconstruction framework proposed in~\cite{pattabiraman2019reconstruction}. Binary polymer-based data storage systems~\cite{Laure2016Coding} operate by designing two molecules of significantly different masses to 
represent the symbols $\{{0,1\}}$ and perform readouts through noisy tandem mass spectrometry. Tandem mass spectrometers fragment the
strings to be read into shorter substrings and only report their masses, often with errors due to imprecise ionization. Modeling the fragmentation 
process output in terms of composition multisets allows for designing asymptotically optimal codes capable of unique reconstruction 
and the correction of a single mass error~\cite{pattabiraman2019reconstruction} through the use of derivatives of Catalan paths. 
Nevertheless, no solutions for multiple-mass error-corrections are currently known. Our work addresses this issue by describing the 
first multiple-error correction codes that use the polynomial factorization approach for the Turnpike problem~\cite{skiena1990reconstructing} and the related factorization described in~\cite{acharya2014string}. Adding Reed-Solomon type coding redundancy into the corresponding polynomials allows for correcting
$t$ mass errors in polynomial time using $t^2\, \log\,k$ redundant bits, where $k$ is the information string length. 
The redundancy can be improved to $\log\,k + t$. However, no decoding algorithm that runs polynomial-time in both $t$ and $n$ for this scheme are currently known, where $n$ is the length of the coded string. 
\end{abstract}

\section{Introduction}

To address the issue of massive data storage, several molecular storage paradigms have recently been put forward in~\cite{al2017mass,goldman2013towards,grass2015robust,yazdi2015rewritable,yazdi2017portable,tabatabaei2019dna}. Among these methods, synthetic polymer-based storage offers the highest promise of low cost and low readout latency~\cite{al2017mass}. In synthetic polymer storage systems,  the two bits $0$ and $1$ are represented by polymers of different masses that are linked through automated phosphoamidite chemistry in a user-specified manner. The stored data is read using tandem mass (MS/MS) spectrometers which provides estimates of the masses of the fragmented polymer. 

Most MS/MS readout systems produce masses of prefixes and suffixes of the data string, which if recovered reliably allow for straightforward string reconstruction. 
Unfortunately, the MS/MS readout process suffers from large mass read error-rates that arise due to imprecise fragmentation. Similar mass error as well as unique reconstruction issues arise in systems that provide the masses of all substrings of the recorded string. 

To address the latter issue, the authors of~\cite{acharya2014string} introduced the problem of \emph{binary string reconstruction from its substring composition multiset}. The substring composition multiset of a binary string is obtained by writing out all substrings of the string of all possible lengths and then representing each substring by its composition. As an example, the string $100$ contains three substrings of length one - $1$, $0$, and $0$, two substrings of length 2 - $10$ and $00$, and one substring of length three - $100$. The composition multiset of the substrings of length one, two and three equals $\{ 0,0,1 \}$, $\{0^11^1,0^2\}$ and $\{0^21^1\}$, respectively. Note that composition multisets ignore information about the actual order of the bits and the substrings and may hence be seen as only capturing the information about the ``mass'' or ``weight'' of unordered substrings. Furthermore, the multiset information cannot distinguish between a string and it's reversal, as well as some other nontrivial interleaved string structures. The problem addressed in~\cite{acharya2014string} was to determine for which string lengths can one guarantee unique reconstruction from an error-free composition multiset up to string reversal. The main results of~\cite[Theorem ~17, ~18, ~20]{acharya2014string} assert that binary strings of length  $\leq 7$, one less than a prime, or one less than twice a prime are uniquely reconstructable up to reversal.

Unlike the work in~\cite{acharya2014string}, the follow-up work of~\cite{pattabiraman2019reconstruction} focused on the problem of constructing uniquely reconstructible strings and uniquely reconstructable strings capable of correcting a single mass error. Both lines of work used the simplifying assumptions that 
one can infer the composition of a fragment polymer from its mass and that when a polymer block is broken down for mass spectrometry analysis, we observe the masses of all its substrings with identical frequency. 

We extend the above described coded string reconstruction study by proposing the first known coding scheme capable of correcting arbitrary multiple mass errors in the polymer strings. Unlike the single-error correction setting which interleaves Catalan-Bertrand paths to obtain codewords with the desired properties we use the two-variate polynomial characterization of the strings first described in~\cite{acharya2014string}. By forcing the polynomials to have specific evaluations at a selected set of elements of an appropriate finite field, we arrive at a Reed-Solomon like characterization of the codestrings. This construction has redundancy $t^2 \log\, k$ bits and also allows for simple polynomial time decoding based on existing Reed-Solomon decoders. We also briefly describe how to extend the 
Catalan-Bertrand framework~\cite{pattabiraman2019reconstruction} for the case of multiple mass errors. For this formulation, the redundancy equals 
$\log\, k + t$ bits while the worst case decoding complexity is exponential in $t$. It remains an open problem to find efficient decoders for this class of codes. Both results add to the growing list of uncoded and coded string reconstruction problems~\cite{levenshtein2001efficient,dudik2003reconstruction,batu2004reconstructing, viswanathan2008improved,kiah2016codes,gabrys2018unique,cheraghchi2019coded}.

\section{Problem Formulation}

Let $\textbf{s}=s_1 s_2 \ldots s_k$ be a binary string of length $k \geq 2$. A substring of $\textbf{s}$ starting at $i$ and ending at $j$, where $1 \leq i < j \leq k,$ is denoted by $\textbf{s}_{i}^{j}$, and is said to have \emph{composition} $0^{z}1^{w}$, where $0 \leq z,w \leq j-i+1$ stand for the number of $0$s and $1$s in the substring, respectively. Note that the composition only conveys information about the weight of the substring, but not the particular order of the bits. Furthermore, let $C_{l}(\textbf{s})$  stand for the multiset of compositions of substrings of $\textbf{s}$ of length $l$, $1\leq l \leq k$. This multiset contains $k-l+1$ compositions. 
The multiset $C(\textbf{s})=\cup_{l=1}^{k} C_{l}(\textbf{s})$ is termed the \emph{composition multiset}. It is straightforward to see that the composition multisets 
of a string $\textbf{s}$ and its reversal, $\textbf{s}^r=s_k s_{k-1} \ldots s_1$ are identical and hence these two strings are indistinguishable based on $C(\cdot)$. 
If a collection of codestrings has the property that all pairs of strings are distinguishable based on their multiset composition, the underlying codebook is referred to as a \emph{reconstruction code}~\cite{pattabiraman2019reconstruction}.

We also define the \emph{cummulative weight} of a composition multiset $C_{l}(\textbf{s}),$ with compositions of the form $0^{z}1^{w}$, where $z+w=l$, as $w_{l}(\textbf{s})=\sum_{0^{z}1^{w} \in C_{l}(\textbf{s})}\, w.$ Observe that $w_{1}(\textbf{s})=w_{k}(\textbf{s})$, as both equal the weight of the string $\textbf{s}$. More generally, one has $w_{l}(\textbf{s})=w_{k-l+1}(\textbf{s}), \text{ for all } 1 \leq l \leq k.$

In our derivations we also make use of the following notation. For a string $\textbf{s}=s_1 s_2 \ldots s_k$, we let $\sigma_i=\text{wt}(s_is_{k-i+1})$ for $i \leq \lfloor \frac{k}{2} \rfloor,$ and $\sigma_{ \lceil \frac{k}{2} \rceil}=\text{wt}(s_{ \lceil \frac{k}{2} \rceil})$, where $\text{wt}$ stands for the weight of the string. 
We also use $\Sigma^{ \lceil \frac{k}{2} \rceil}$ to denote the sequence $( \sigma_i)_{i \in [ \lceil \frac{k}{2} \rceil]},$ where $[a]=\{{1,\ldots,a\}}$. Whenever clear from the context, we omit the argument $\textbf{s}$ and the floors/ceiling functions required to obtain appropriate integer lengths.

We now describe our problem setup. One is given a valid composition multiset of a string $\textbf{s}$, $C(\textbf{s})$. Within the multiset $C(\textbf{s})$, some compositions may be arbitrarily corrupted. We refer to such errors as \textbf{composition errors}. For example, when $\textbf{s}=100101$, the multiset 
$C_2(\textbf{s})=\{{0^11^1,0^2,0^11^1,0^11^1,0^11^1\}}$ may be corrupted to $\hat{C}_2(\textbf{s})=\{{\mathbf{0^2},0^2,0^11^1,0^11^1,0^11^1\}}$, in which case we have a single composition error. Furthermore, the multisets $C_2(\textbf{s})$ and $C_5(\textbf{s})$ may be corrupted to $\hat{C}_2(\textbf{s})=\{{\mathbf{0^2},0^2,0^11^1,0^11^1,0^11^1\}}$ and $\hat{C}_5(\textbf{s})=\{{\mathbf{0^11^4},0^31^2\}}$, in which case we say that we encountered an example of two composition errors. 

The problem at hand is to design the largest reconstructable codebook of strings with $k$ information bits and of length $n$ such that any $t < n$ composition errors can be correctly identified and corrected. 

\section{Main Results: Error-Correcting Reconstruction Codes}\label{sec:symmetric}

We now turn our attention to reconstruction codes capable of correcting multiple composition errors. The proposed method leverages a polynomial formulation of the composition reconstruction problem first described in~\cite{acharya2014string}. The main result is a constructive proof for the existence of codes with $\O(t^2 \log n)$ bits of redundancy capable of correcting $t$ composition errors.

To this end, we first review the results of~\cite{acharya2014string} that describe the string reconstruction problem using bivariate polynomial factorization. For a string $\textbf{s} \in \{0,1\}^n$, let $P_{\textbf{s}}(x,y)$ be a bivariate polynomial of degree $n$ with coefficients in $\{0,1\}$ such that $P_{\textbf{s}}(x,y)$ contains exactly one term with total degree $i \in \{0,1,\ldots, n\}$. If $\textbf{s}= s_1 \ldots s_n$ and if $\Big (P_{\textbf{s}}(x,y) \Big)_i$ denotes the unique term of total degree $i$, then $\Big (P_{\textbf{s}}(x,y) \Big)_0 = 1$, and 
\begin{align*}
\Big (P_{\textbf{s}}(x,y) \Big)_i = \begin{cases}
y \, \Big(P_{\textbf{s}}(x,y) \Big)_{i-1}, \text{ if $s_i = 0$}, \\
x \, \Big(P_{\textbf{s}}(x,y) \Big)_{i-1}, \text{ if $s_i = 1$.}
\end{cases}
\end{align*}

In words, we use $y$ to denote the bit $0$ and $x$ to denote the bit $1$ and then summarize the composition of all prefixes of the string $\textbf{s}$ in polynomial form. As a simple example, for $\textbf{s} = 0100$ we have $P_{\textbf{s}}(x,y) = 1 + y + xy + xy^2 + xy^3$: We start with the free coefficient $1$, then add $y$ to indicate that the prefix of length one of the string equals $0$, add $xy$ to indicate that the prefix of length two contains one $0$ and one $1$, add $xy^2$ to indicate that the prefix of length three contains two $0$s and one $1$ and so on. 

We also introduce another bivariate polynomial $S_{\textbf{s}}(x,y)$ to describe the composition multiset $C(\textbf{s})$ in a manner similar to $P_{\textbf{s}}(x,y)$. In particular, we now associate each composition with a monomial in which the symbol $y$ represents the bit $0$ and the symbol $x$ with the bit $1$. As an example, for $\textbf{s} = 0100$ we have 
$C(\textbf{s})=\Big \{ 0 , 1, 0, 0, 01, 01, 0^2, 0^21, 0^21, 0^31 \Big \},$
and 
$S_{\textbf{s}}(x,y) = x + 3y + 2xy + y^2 + 2xy^2 + xy^3,$
where the first two terms in $S_{\textbf{s}}(x,y)$ indicate that the composition multiset contains one substring $1$ and three substrings $0$; the next three terms indicate that the string contains two substrings with one $1$ and one $0$ and one substring with two $0$s. The remaining terms are interpreted similarly. 

The key identity observation from~\cite{acharya2014string} is as follows:
\begin{align}\label{eq:SP}
P_{\textbf{s}} (x,y) \, P_{\textbf{s}}\left(\frac{1}{x}, \frac{1}{y}\right) =  (n+1) + S_{\textbf{s}}(x,y) + S_{\textbf{s}}\left(\frac{1}{x}, \frac{1}{y}\right).
\end{align}
Given a bivariate polynomial $f(x,y)$, we use $f^{*}(x,y)$ to denote its reciprocal polynomial, defined as
$$ f^{*}(x,y) = x^{\text{deg}_x(f)} y^{\text{deg}_y(f)} f\left(\frac{1}{x}, \frac{1}{y}\right),$$
where $\text{deg}_x(f)$ denotes the $x$-degree of $f(x,y)$ and $\text{deg}_y(f)$ denotes its $y$-degree. For simplicity, we hence write $d_x = \text{deg}_x(P_{\textbf{s}})$ and $d_y = \text{deg}_y(P_{\textbf{s}})$. Using the notion of the reciprocal polynomial we can rewrite the expression in (\ref{eq:SP}) as:
\begin{align}
P_{\textbf{s}} (x,y) \, P^{*}_{\textbf{s}}(x, y) = x^{d_x} y^{d_y} \, \left( n + 1 + S_{\textbf{s}}(x,y) \right) + S^{*}_{\textbf{s}}(x,y).
\end{align}

Note that if $C'(\textbf{s})$ is the composition multiset resulting from $t$ composition errors in $C(\textbf{s})$ and $\tilde{S}_{\textbf{s}}(x,y)$ is the polynomial representation for $C'(\textbf{s})$ while $S_{\textbf{s}}(x,y)$ is the polynomial representation for $C(\textbf{s})$, then we have:
\begin{align*}
\tilde{S}_{\textbf{s}}(x,y) = S_{\textbf{s}}(x,y) + E(x,y),
\end{align*}
where $E(x,y)$ has at most $2t$ terms. Our first result relates $\tilde{S}_{\textbf{s}}(x,y)$ and $P_{\textbf{s}}(x,y)$. 

\begin{claim}\label{cl:equiv} Suppose that $\emph{wt}(\textbf{s}) \bmod 2t+1 \equiv c_w$ for some $c_w \in \{0,1,\ldots, 2t\}$. Then, given $\tilde{S}_{\textbf{s}}(x,y)$ and $c_w$ one can generate
\begin{align*}
P_{\textbf{s}}(x,y) \, P^{*}_{\textbf{s}}(x,y)  + \tilde{E}(x,y),
\end{align*}
where the polynomial $\tilde{E}(x,y)$ has at most $4t$ terms.
\end{claim}
\begin{proof} First, recall that $\tilde{S}_{\textbf{s}}(x,y) = S_{\textbf{s}}(x,y) + E(x,y)$ where $E(x,y)$ has at most $2t$ terms. Given $c_w$, we can easily determine the degrees $d_x$ and $d_y$ of the polynomial encoding of $\textbf{s}$. Next, we form $P_{\textbf{s}}(x,y) \, P^{*}_{\textbf{s}}(x,y)$ as follows:
\begin{align*}
&x^{d_x} y^{d_y} \, \left( n + 1 + \tilde{S}_{\textbf{s}}(x,y) + \tilde{S}_{\textbf{s}} \left(\frac{1}{x}, \frac{1}{y}\right) \right) \\
&= x^{d_x} y^{d_y} (n+1) + x^{d_x} y^{d_y} \, \times \\
   &\left(S_{\textbf{s}}(x,y) + E(x,y) 
  +  S_{\textbf{s}}\left(\frac{1}{x}, \frac{1}{y}\right) + E\left(\frac{1}{x},\frac{1}{y}\right) \right)   \\
&= P_{\textbf{s}}(x,y) \, P^{*}_{\textbf{s}}(x,y) + x^{d_x} y^{d_y} \left( E(x,y) + E \left( \frac{1}{x}, \frac{1}{y} \right) \right) \\
&= P_{\textbf{s}}(x,y) \, P^{*}_{\textbf{s}}(x,y) + \tilde{E}(x,y),
\end{align*}
where $\tilde{E}(x,y) = x^{d_x} y^{d_y} \left( E(x,y) + E \left( \frac{1}{x}, \frac{1}{y} \right) \right)$ has at most $4t$ nonzero terms, which proves the desired result.
\end{proof}
%
%

Let $\mathbb{F}_q$ be a finite field of order $q$, where $q$ is an odd prime. Let $\alpha \in \mathbb{F}_q$ be a primitive element of the field. For a polynomial $f(x) \in \mathbb{F}_q[x]$, let $\R(f)$ denote the set of its roots. We find the following result useful for our subsequent derivations.

\begin{theorem}\label{th:RS} (\cite[Ch.~5]{roth2006introduction}) Assume that $E(x) \in \mathbb{F}_q[x]$ has $\leq t$ terms. Then, $E(x)$ can be uniquely determined in $\cO(n^2)$ time given $E(\alpha^t), E(\alpha^{t-1}), \ldots, E(\alpha^0), E(\alpha^{-1}), \ldots, E(\alpha^{-t})$.
\end{theorem}

\subsection{The Code Construction}
Our approach to constructing a $t$-error-correcting code of length $n$, denoted by $\mathcal{S}^{(t)}_{E}(n)$, relies on the fact that $\tilde{E}(x,y)$ may be written as:
\begin{align}\label{eq:exy}
\tilde{E}(x,y) =& ( a_{i_{1},1} y^{j_{i_1,1}} + \cdots + a_{i_1,m_{i_1}} y^{j_{i_1,m_{i_1}}}) x^{i_1} + \nonumber \\
&(a_{i_2,1} y^{j_{i_2,1}} + \cdots + a_{i_2,m_{i_2}} y^{j_{i_2,m_{i_2}}}) x^{i_2} +  \nonumber \\
& \qquad \qquad \quad \quad \, \vdots \\
&(a_{i_h,1} y^{j_{i_h,1}} + \cdots + a_{i_h,m_{i_h}} y^{j_{i_h,m_{i_h}}}) x^{i_h}, \nonumber
\end{align}
where each $a_{i,j} \in \{-1,1\}$, $h \leq 4t$ and the total number of nonzero terms is $\leq 4t$. Since $\tilde{E}(x,y)$ is restricted to have at most $4t$ nonzero terms, each of the polynomials $(a_{i_\ell,1} y^{j_{i_\ell,1}} + \cdots + a_{i_\ell,m_{i_\ell}} y^{j_{i_\ell,m_{i_\ell}}})$ can contain at most $4t$ nonzero terms. Consequently, one has $m_{i_\ell} \leq 4t$ for all $\ell \in \{1,2,\ldots,h\}$. 

Based on the previous observations we are ready to introduce our first code construction described in the lemma that follows. Henceforth, we assume that $P_{\textbf{s}}(x,y)$ is a bivariate polynomial over the field $\mathbb{F}_q$ where $q=2n+1$ is an odd prime. Clearly, for a $P_{\textbf{s}}(x,y) \in \mathbb{I}[x,y]$ over the integers, one can obtain $P_{\textbf{s}}(x,y) \in \mathbb{F}_q[x,y]$ by simply applying the modulo $q$ operation on $P_{\textbf{s}}(x,y)$.

\begin{lemma}\label{lem:deca1} Let $\C \subseteq \{0,1\}^n$ be a collection of strings $\textbf{s}$ that satisfy
\begin{align*}
\emph{wt}(\textbf{s}) \bmod 2t+1 &= 0, \\
\{1,\alpha,\alpha^2, \ldots, \alpha^{4t} \} &\subseteq \R(P_{\textbf{s}}(x,1)),\\
\{1,\alpha,\alpha^2, \ldots, \alpha^{4t} \} &\subseteq \R(P_{\textbf{s}}(x,\alpha)),\\
                                            & \, \vdots \\
\{1,\alpha,\alpha^2, \ldots, \alpha^{4t} \} &\subseteq \R(P_{\textbf{s}}(x,\alpha^{4t})).
\end{align*}
Then, $\C$ is a $t$-error-correcting code.
\end{lemma}
\begin{proof} We prove the claim by describing a decoding algorithm that for any given $\tilde{S}_{\textbf{s}}(x,y)$, which is the result of at most $t$ composition errors occurring in $S_{\textbf{s}}(x,y)$, uniquely recovers $S_{\textbf{s}}(x,y)$.

Since there are at most $t$ erroneous compositions in $\tilde{S}_{\textbf{s}}(x,y)$, one can determine $\text{wt}(\textbf{s})$ by summing up the length-one compositions (i.e., the bits) in $\tilde{S}_{\textbf{s}}(x,y)$ along with the fact that $\text{wt}(\textbf{s}) \bmod 2t+1 = 0$. Therefore, from Claim~\ref{cl:equiv}, we can construct the polynomial
\begin{align}\label{eq:l1deceq}
F(x,y) = P_{\textbf{s}}(x,y) \, P^{*}_{\textbf{s}}(x,y)  + \tilde{E}(x,y),
\end{align} 
where $\tilde{E}(x,y)$ has at most $4t$ nonzero terms. Suppose that $\beta, \beta' \in \mathbb{F}_q$. First, observe that if $P_{\textbf{s}}(\beta,\beta') \, P^{*}_{\textbf{s}}(\beta, \beta') = 0$, then $P_{\textbf{s}}(\frac{1}{\beta},\frac{1}{\beta'}) \, P^{*}_{\textbf{s}}(\frac{1}{\beta},\frac{1}{\beta'}) = 0$ which immediately follows from the definition of $P^{*}_{\textbf{s}}(x,y)$. Since $\{1,\alpha,\alpha^2, \ldots, \alpha^{4t} \} \subseteq \R(P_{\textbf{s}}(\alpha^{\ell_1},y))$ for all $\ell_1 \in \{ 0,1,\ldots,4t \}$, and similarly $\{1,\alpha,\alpha^2, \ldots, \alpha^{4t} \} \subseteq \R(P_{\textbf{s}}(x,\alpha^{\ell_2}))$ for all $\ell_2 \in \{ 0,1,\ldots,4t \},$ it follows that $F(\alpha^{\ell_1},\alpha^{\ell_2})=\tilde{E}(\alpha^{\ell_1},\alpha^{\ell_2})$. Hence, we have:
\begin{align*}
\tilde{E}(\alpha^{\ell_1},&\alpha^{\ell_2}) = \\
& \quad \, \, \big( a_{{i_1,1}} \alpha^{{\ell_2} \times j_{i_1,1}} + \cdots + a_{i_1,m_{i_1}} \alpha^{{\ell_2} \times j_{i_1,m_{i_1}}} \big ) \alpha^{\ell_1 \times i_1}\\
&+ \big ( a_{{i_2,1}} \alpha^{{\ell_2} \times j_{i_2,1}} + \cdots + a_{{i_2,m_{i_2}}} \alpha^{{\ell_2} \times j_{i_2,m_{i_2}}} \big) \alpha^{\ell_1 \times i_2}  \\
&  \qquad \qquad \qquad \qquad \, \quad \vdots  \\ 
&+ \big( a_{{i_h,1}} \alpha^{{\ell_2} \times j_{i_h,1}} + \cdots + a_{{i_h,m_{i_h}}} \alpha^{{\ell_2} \times j_{i_h,m_{i_h}}} \big) \alpha^{\ell_1 \times i_h},
\end{align*}
for ${\ell_1},{\ell_2} \in  \{ 0,1,\ldots,4t,-1,-2,\ldots,-4t \}$. From Theorem~\ref{th:RS}, for any fixed $\ell_2$ we know the evaluations $\tilde{E}(\alpha^{\ell_1}, \alpha^{\ell_2})$ for $\ell_1 \in \{0,1,\ldots,4t,-1,-2, \ldots,-4t\}$, so that we can recover the following polynomials:
\begin{align}\label{eq:l1efy}
\tilde{E}(x,\alpha^{\ell_2}) &= \big( a_{{i_1,1}} \alpha^{{\ell_2} \times j_{i_1,1}} + \cdots + a_{i_1,m_{i_1}} \alpha^{{\ell_2} \times j_{i_1,m_{i_1}}} \big ) x^{i_1}  \nonumber \\
&+ \big ( a_{{i_2,1}} \alpha^{{\ell_2} \times j_{i_2,1}} + \cdots + a_{{i_2,m_{i_2}}} \alpha^{{\ell_2} \times j_{i_2,m_{i_2}}} \big) x^{i_2}  \nonumber  \\
&\qquad \qquad \qquad \qquad \, \quad \vdots \nonumber \\
&+ \big( a_{{i_h,1}} \alpha^{{\ell_2} \times j_{i_h,1}} + \cdots + a_{j_{i_h,m_{i_h}}} \alpha^{{\ell_2} \times j_{i_h,m_{i_h}}} \big) x^{i_h},
\end{align} 
using the decoder for a cyclic Reed-Solomon code, which has complexity $\cO(n^2)$. 

Let 
$$M_{i_\ell}(y) =  a_{{i_\ell,1}} y^{j_{i_\ell,1}} + \cdots + a_{i_\ell,m_{i_{\ell}}} y^{j_{i_\ell,m_{i_\ell}}}$$
be the polynomial multiplier of $x^{i_\ell}$ in $\tilde{E}(x,y)$. From the previous discussion, we know that the maximum number of nonzero terms in $M_{i_\ell}(x)$ is $4t$. Using (\ref{eq:l1efy}), we can determine $M_{i_\ell}(\alpha^{\ell_2})$ for $\ell_2 \in \{0,1,2,\ldots,4t, -1, -2, \ldots, -4t\}$. Due to Theorem~\ref{th:RS}, this implies that we can recover $M_{i_\ell}(y)$ for $\ell\in \{1,2,\ldots, h \}$ once again using a decoder for a Reed-Solomon code. Since $\tilde{E}(x,y) = M_{i_1}(y) x^{i_1} + M_{i_2}(y) x^{i_2} + \cdots + M_{i_h}(y) x^{i_h}$, we can determine $E(x,y)$ and subsequently reconstruct $S_{\textbf{s}}(x,y)$ given $\tilde{S}_{\textbf{s}}(x,y)$.
\end{proof}

The following corollary follows immediately from Lemma~\ref{lem:deca1}. 

\begin{corollary}\label{cor:sideinfo} Let $\C \in \{0,1\}^n$ be a collection of strings $\textbf{s}$ that satisfy
\begin{align*}
P_{\textbf{s}}(\alpha^{\ell_1}, \alpha^{\ell_2}) = a_{\ell_1,\ell_2} \text{  and  } \emph{wt}(\textbf{s}) \equiv a \bmod 2t+1,
\end{align*}
for all $\ell_1,\ell_2 \in \{0,1,\ldots, 4t\}$, and where $(a_{\ell_1,\ell_2})_{\ell_1=0, \ell_2=0}^{4t}$ is an arbitrary vector from $\mathbb{F}_q^{ (4t+1)^2}$ and $a \in \{0,1,\ldots, 2t+1\}$. Then, $\C$ corrects $t$ composition errors. 
\end{corollary}

\subsection{A Systematic Encoder $\E_{t,n}$}

We construct next a systematic encoder for the previously proposed codes. The focus is on a systematic encoder $\E_{t,n}.$ 

Let $r$ be the number of redundant bits in the proposed code construction. We will show in Theorem~\ref{th:mainth} that for all $n$, one has 
\begin{align*}
r \leq&\ 4 \Big[ (4t+1)^2 (\log (2n+1)+1) + \log (2t+1) \\
&+ t \left( \log (4t+1)^2 (\log (2n+1)+1) + \log (2t+1) \right) \Big] \\
&+ \frac{1}{2} \log(n).
\end{align*}
One can show that $r \leq 156 t^2 \log 8n$.
Thus, $r= \cO(t^2 \log n)$. Furthermore, $r \leq 156 t^2 \log 8k + 156 t^2 \left( \frac{1}{\kappa} \right)$, where $\kappa$ is supremum over all $\kappa > 0$ such that $n \geq (1+\kappa) 156 t^2 \log 8n$.

The encoder $\E_{t,n}$ takes as input the string $\textbf{u} \in \{0,1\}^{n-\hat{r}}$, where $\hat{r}>0$ is a redundancy to be specified in what follows, and it produces a string $\textbf{s}$. Note that the evaluations of the polynomial $P_{\textbf{s}}(x,y)$ are stored in vector-form 
$$ \left( w_1, w_2, \ldots, w_{\frac{\hat{r}}{2}} \right) \bmod 2,$$
where the cummulative weights $w_i$s of a composition multiset $C_{i}$ are as defined at the beginning of Section II. 

Let $\E_t : \{0,1\}^{m} \to \{0,1\}^{m + t \log m}$ be a systematic encoder for a code with minimum Hamming distance $2t+1$ that inputs a string of length $m$ and outputs a string of length $m+t\log m$. We will use this encoder with $m= (4t+1)^{2}+1$. The encoder inputs $\textbf{u} \in \{0,1\}^{n-\hat{r}}$ and outputs $\textbf{s} \in \{0,1\}^n$ while executing the following steps.
\vspace{3pt}
\hrule \vspace{0.5pt} \hrule
\vspace{3pt}
\textbf{Encoder} $\E_{t,n} : \{0,1\}^{n- \hat{r}} \to \{0,1\}^n$.
\vspace{3pt}
\hrule \vspace{0.5pt} \hrule
\vspace{3pt}
\textbf{Input} String $\textbf{u} \in \{0,1\}^{n- \hat{r}}$. 

\textbf{Output} Codestring $\textbf{s} \in \{0,1\}^{n}$ that corrects $t$ errors.
\vspace{3pt}
\hrule 
\vspace{3pt}
\begin{enumerate}
\item Let $\alpha \in \mathbb{F}_q$ be a primitive element and $q$ be an odd prime $\geq 2n+1$. For $\ell_1,\ell_2 \in \{0,1,\ldots, 4t\}$, set $a_{\ell_1,\ell_2} = P_{\textbf{u}}(\alpha^{\ell_1}, \alpha^{\ell_2})$, $\ba = (a_{\ell_1,\ell_2})_{\ell_1=0, \ell_2=0}^{4t}$. \\
Let $a=\text{wt}(\textbf{u}) \bmod 2t+1$. 
\item Let $\bf{\bar{s}} = \E_t($$a$, $\ba) \in \{0,1\}^{\frac{\hat{r}}{4}}$.
\item For $j \in \{1,2, \ldots, \frac{\hat{r}}{2} \}$, define $\textbf{z}=(z_1 \ldots z_{\frac{\hat{r}}{2}})$ as 
\begin{align*}
z_{j} = \begin{cases}
 \sum_{i=1}^{j-1} z_{i} \bmod 2, &\text{ if $j$ is odd and } \bar{s}_{\frac{j+1}{2}}=0,\\
 \sum_{i=1}^{j-1} z_{i} + 1 \bmod 2, &\text{ if $j$ is odd and } \bar{s}_{\frac{j+1}{2}} = 1, \\
 0, &\text{ if $j$ is even.}
\end{cases}
\end{align*}
\item Set $\textbf{s} =  \textbf{0} \, \textbf{u} \, \textbf{z}  \in \{0,1\}^n$, where $\textbf{0}$ is an all-zero string of length $\frac{\hat{r}}{2}$.
\end{enumerate}
\vspace{3pt}
\hrule \vspace{0.5pt} \hrule
\vspace{3pt}
The $t$-error-correcting code $\S_{E}^{(t)}(n)$ is generated by the following two-step procedure:
\begin{itemize}
\item An information string of length $k$ is first encoded using the reconstruction code described in~\cite{pattabiraman2019reconstruction}, resulting in the string $\textbf{u} \in \S_{R}(n - \hat{r})$, where $\S_{R}(n - \hat{r})$ stands for the underlying reconstruction code.
\item The string $\textbf{u}$ is passed through the encoder $\E_{t,n}$, resulting in the codestring $\textbf{s}  = \E_{t,n} (\textbf{u}) \in \S_{E}^{(t)}(n)$. 
\end{itemize}
Consequently, we should have $\hat{r} = r- \left( \frac{1}{2} \log(n) \right)$. 

Thus, the number of redundancy bits is calculated as follows: 1) Since $\mathbb{F}_q$ is over a prime $q \geq 2n+1$, every $\alpha_{\ell_1,\ell_2}$, $\ell_1, \ell_2 \in \{0,1, \dots 4t \}$ requires at most $1 + \log (2n + 1)$ (as given any positive integer $x$, there exits a prime number between $x$ and $2x$). 2) Note that $a$ requires $\log 2t +1$. Thus, $\frac{\hat{r}}{4}$ is at most $(4t+1)^2 (1+\log (2n+1)) + \log (2t+1) + t \log ((4t+1)^2 (1+\log (2n+1)) + \log (2t+1)) $. 3) As mentioned earlier, the reconstruction string $\textbf{u}$ requires $r \leq \frac{1}{2}\log n$ redundancy bits. Thus, the encoder $\E_{t,n}$ requires $\O (t^2 \log n)$ additional bits. 

We find the following claims useful in our subsequent derivations.

\begin{claim}\label{cl:zr} At Step 3) of the encoding procedure, for odd $j \in [\frac{\hat{r}}{2}]$, one has $\bar{s}_{\frac{j+1}{2}} = \sum_{i=1}^j z_i \bmod 2.$
\end{claim}
This claim obviously follows from the definition of the string $\textbf{z}$.

Recall next that for a string $\textbf{s} \in \{0,1\}^n$, its $\Sigma^{n/2}$ sequence $( \sigma_1, \sigma_2, \ldots, \sigma_{\frac{n}{2}}) \in \{0,1,2\}^{\frac{n}{2}}$ equals $\sigma_i = s_i + s_{n+1-i}$. As a result of Step 4) of encoding with $\E_{t,n}$, we have the next claim.

\begin{claim}\label{cl:rs} For $j \in [\frac{\hat{r}}{2}]$,
\begin{align*}
z_j = \sigma_j.
\end{align*}
\end{claim}

The next claim connects the quantities $w_i$ and $\bar{\textbf{s}}$, defined in Step 2 of the encoding procedure.

\begin{claim}\label{cl:rw} For $j \in [\frac{\hat{r}}{4}]$, it holds
$$ w_{2j} \mod 2 = \bar{s}_{j}. $$
\end{claim}
\begin{proof} The result follows by noting that
\begin{align*}
w_{2j} \equiv\ &2j w_1 - (2j-1) \sigma_1 - (2j-2) \sigma_2- \cdots - \sigma_{2j-1} \bmod 2 \\
\equiv\ &\sigma_1 + \sigma_3 + \cdots + \sigma_{2j-1} \bmod 2,
\end{align*}
where the first line follows from the fact that $$\frac{1}{i} \sigma_1 + \frac{2}{i} \sigma_2 + \dots + \frac{i-1}{i} \sigma_{i-1} +  \sigma_i + \sigma_{i+1} + \dots +  \sigma_{n/2} = \frac{1}{i} w_{i}. $$ From Claims~\ref{cl:zr} and \ref{cl:rs}, and the previous observation, and along with the fact that $z_j = 0$ for even values of $j$ in Step 3) of the encoding, we have
\begin{align*}
w_{2j} \equiv \sum_{i=1}^{2j-1} \sigma_j \equiv \sum_{i=1}^{2j-1} z_j \equiv \bar{s}_{j} \bmod 2.
\end{align*}
\end{proof}
The following result will be used to prove the main finding regarding the error-correction, as stated in Theorem~\ref{th:mainth}.
\begin{lemma}\label{lem:mclemma} The code defined as 
\begin{align*}
\C = \Big \{ \textbf{s} : \textbf{s} = \E_{t,n}(\textbf{u}), \textbf{u} \in \{0,1\}^{n - \hat{r}} \Big \}.
\end{align*}
is a $t$-error-correcting code.
\end{lemma}
\begin{proof} In order to prove the result, we will describe how to recover $S_{\textbf{s}}(x,y)$ given $\tilde{S}_{\textbf{s}}(x,y),$ where $\tilde{S}_{\textbf{s}}(x,y)$ is the result of at most $t$ composition errors in $S_{\textbf{s}}(x,y)$ for a codestring generated as $\E_{t,n}(\textbf{u}) = \textbf{s}$. We begin by forming the string $\tilde{\bw} = \Big ( \tilde{w}_{2}, \tilde{w}_{4}, \ldots, \tilde{w}_{\frac{\hat{r}}{4}}  \Big).$
This vector is obtained from $\tilde{S}_{\textbf{s}}(x,y)$ by summing up the ones in all compositions of length two to get $\tilde{w}_2$, summing up the ones in all compositions of length four to get $\tilde{w}_4$, and so on. Let $\bw = \Big ( w_{2}, w_{4}, \ldots, w_{\frac{\hat{r}}{4}}  \Big)$ for the string $\textbf{s}$. 

Since there are at most $t$ composition errors in $\tilde{S}_{\textbf{s}}(x,y)$, it follows that $d_H \Big( {\bw} \bmod 2, \tilde{{\bw}} \bmod 2 \Big) \leq t.$
From Claim~\ref{cl:rw}, since $\bw \bmod 2$ belongs to a code with minimum Hamming distance $2t+1$, we can recover $\bw \bmod 2$ from $\tilde{\bw} \bmod 2$. Then, given $\bw \bmod 2,$ we can recover $\bar{\textbf{s}}$ from Step 2) of the encoding procedure, and from $\bar{\textbf{s}}$ we can determine $a = \text{wt}(\textbf{u})$. 
Using $\bar{\textbf{s}}$, it is also straightforward to determine $\textbf{z}$ from Step 3) of the encoding procedure. Subsequently, we can recover $
\text{wt}(\textbf{s}) = a + \text{wt}(\textbf{u}),$
and from $\text{wt}(\textbf{s})$, we can determine $d_x$ and $d_y$, the $x$ and $y$ degrees of the polynomial $P_{\textbf{s}}(x,y)$. 

Next, we turn our attention to recovering the evaluations of the polynomial $P_{\textbf{s}}(\alpha^{\ell_1},\alpha^{\ell_2})$ for $\ell_1,\ell_2 \in \{0,1,\ldots, 4t\}$. These, along with $\text{wt}(\textbf{s})$, suffice according to Lemma~\ref{lem:deca1} to recover $\textbf{s}$.
From $
\bar{\textbf{s}} $, we can determine $P_{\textbf{u}}(\alpha^{\ell_1}, \alpha^{\ell_2})$ according to Steps 1) and 2) of the encoding procedure. 

Let $d_{x,\textbf{u}} = \deg_x(P_{\textbf{u}}(x,y))$ and $d_{y,\textbf{u}} = \deg_y (P_{\textbf{u}}(x,y))$. 

First, note that
\begin{align*}
P_{\textbf{s}}(x,y) &=P_{\textbf{0}}(x,y)+ y^{\frac{\hat{r}}{2}} (P_{\textbf{u}}(x,y)-1) \\
&+ x^{d_{x,\textbf{u}}} y^{\frac{\hat{r}}{2}+d_{y,\textbf{u}}} \, (P_{\textbf{z}}(x,y)-1).
\end{align*}
Therefore, we can recover $P_{\textbf{s}}(\alpha^{\ell_1},\alpha^{\ell_2})$ using
\begin{align*}
P_{\textbf{s}}(\alpha^{\ell_1},\alpha^{\ell_2}) &= P_{\textbf{0}}(\alpha^{\ell_1},\alpha^{\ell_2}) + \alpha^{\ell_2 \times \frac{\hat{r}}{2}} (P_{\textbf{u}}(\alpha^{\ell_1},\alpha^{\ell_2})-1)\\
&+ \alpha^{\ell_1 \times d_{x,\textbf{u}}} \alpha^{\ell_2 \times (\frac{\hat{r}}{2}+d_{y,\textbf{u}} )} \, (P_{\textbf{z}}(\alpha^{\ell_1},\alpha^{\ell_2})-1),
\end{align*}
since $\textbf{z}$ was already recovered. The proof of the claim now follows from Corollary~\ref{cor:sideinfo}. Error-correction can be performed in $\cO (t n^2)$ time.
\end{proof}

Thus, we are left with the task of reconstructing the string $\textbf{s}$ from its correct composition multiset $C(\textbf{s})$. If all pairs of prefixes and suffixes of the same length are such that their weights differ, the string can be reconstructed efficiently by the non-backtracking algorithm~\cite{pattabiraman2019reconstruction}. Recall that the string $\textbf{s}$ is obtained by concatenating three strings, \textit{i.e.}, $\textbf{s} = \textbf{0} \, \textbf{u} \, \textbf{z}$. The prefix of length $\frac{\hat{r}}{2}$ is fixed to be all zeros and can therefore be reconstructed immediately. Lemma~\ref{lem:mclemma} allows one to recover the suffix $\textbf{z}$. Since $\textbf{u} \in \S_{R}(n - \hat{r})$, any prefix of length $\frac{\hat{r}}{2}+1$ has strictly more $0$s than its corresponding suffix of the same length. Thus, the non-backtracking algorithm reconstructs the correct string $\textbf{s}$ in $\cO(n^3)$ time. This gives rise to the following result.

\begin{theorem}\label{th:mainth} There exists a systematic $t$-error correcting code with redundancy $\cO(t^2 \log k)$ and decoding complexity $\cO(n^3)$.
\end{theorem}

The above result can be improved by using a Catalan path construction akin to the one proposed for single-error correction in~\cite{pattabiraman2019reconstruction}. To this end, let $\C(n) \subset \{0,1\}^n$ denote the set of Catalan paths of length $n$. It is well-known that the code $\C(n)$ has approximately $\log n$ bits of redundancy, which follows directly from their number $\frac{1}{n/2+1}\binom{n}{n/2}$ (where we tacitly assumed that $n$ is even). Let
\begin{align*}
\C(n,t) &= \Big \{ \textbf{s} \in \{0,1\}^n \, : \,  s_1\, s_2 \ldots \, s_{4t+1}= 0 \, 0 \, \ldots 0, \\
& \hspace{1.865cm} s_{n-4t} \, s_{n-4t+1} \ldots s_{n} = 1\,1\ldots \,1, \\
& \hspace{0.3cm} s_{4t+2} \, s_{4t+3} \ldots \, s_{n-4t-1} \in \C(n-2(4t+1)) \Big \}.
\end{align*}
It can be shown that $\C(n,t)$ is a $t$-composition error-correcting code with $\cO(\log n + t)$ bits of redundancy, which represents a significant improvement compared to the previously described construction. The worst-case decoding complexity of the code scales exponentially with $t$.

\section*{Acknowledgment} The work was supported by the NSF Grant 1618366, the SemiSynBio NSF+SRC program under grant number 1807526 and the DARPA Molecular Informatics program.

\IEEEtriggeratref{9}

\bibliography{biblio} 
\bibliographystyle{ieeetr}


\end{document}